\documentclass[11pt,twoside,letterpaper]{llncs}
\usepackage{times}
\usepackage[T1]{fontenc}   
\usepackage{amsmath}
\usepackage{amssymb}
\usepackage{fancybox}
\usepackage{graphics,graphicx}
\newtheorem{fact}{Fact}
\newcommand{\junk}[1]{}

\begin{document}
\title{Efficient Assignment of Identities
  in Anonymous Populations}
%\author{Anonymous authors
  %$\ $
  %\inst{1}
 % \institute{$\ $}}
%\junk{
  \author{Leszek G\k{a}sieniec
  \inst{1}
  \and
  Jesper Jansson
  \inst{2}
  \and
  Christos Levcopoulos
  \inst{3}
\and
\newline Andrzej Lingas
\inst{3}
\institute{Department of Computer Science, University of Liverpool, 
Street, L69 38X, U.K.
  \texttt{L.A.Gasieniec@liverpool.ac.uk}
\and 
 Graduate School of Informatics, Kyoto University, Yoshida-Honmachi, Sakyo-ku,
Kyoto 606-8501, Japan.   
\texttt{jj@i.kyoto-u.ac.jp}
\and
Department of Computer Science, Lund University, 22100 Lund, Sweden. 
\newline \texttt{$\{$Christos.Levcopoulos, Andrzej.Lingas$\}$@cs.lth.se}}
  }
 % }
%\date{}
  \maketitle
\begin{abstract}

We consider the fundamental problem of assigning distinct labels to
agents in the probabilistic model of population protocols. Our
protocols operate under the assumption that the size $n$ of the
population is embedded in the transition function.  Their efficiency
is expressed in terms of the number of states utilized by agents, the
size of the range from which the labels are drawn, and the expected
number of interactions required by our solutions.  Our primary goal is
to provide efficient protocols for this fundamental problem
complemented with tight lower bounds in all the three aspects.
W.h.p. (with high probability) our labeling protocols
are silent, i.e., eventually each agent reaches its
final state and remains in it forever, and they are safe, i.e., never
update the label assigned to any single agent.
We first present a silent w.h.p. and safe labeling protocol that draws labels
from the range $[1,2n].$ Both, the number of interactions required and
the number of states used by the protocol are asymptotically optimal,
i.e., $O(n \log n)$ w.h.p. and $O(n)$, respectively.  Next, we present
a generalization of the protocol, where the range of assigned labels
is $[1,(1+\varepsilon ) n]$. The generalized protocol requires $O(n
\log n / \varepsilon )$ interactions in order to complete the
assignment of distinct labels from $[1,(1+\varepsilon ) n]$ to the $n$
agents, w.h.p.  It is also silent w.h.p. and safe, and uses
$(2+\varepsilon)n+O(n^c)$ states, for any positive $c<1.$ On the other
hand, we consider the so-called pool labeling protocols that include
our fast protocols. We show that the expected number of interactions
required by any pool protocol is $\ge \frac{n^2}{r+1}$, when the
labels range is $1,\dots, n+r<2n.$ Furthermore, we provide a
protocol which uses only $n+5\sqrt n +O(n^c)$ states, for any
$c<1,$ and draws labels from the range $1,\dots,n.$ The expected
number of interactions required by the protocol is $O(n^3).$
Once a unique leader is elected it produces a valid labeling
and it is silent and safe. On the
other hand, we show that
(even if a unique leader is given in advance)
any silent protocol that produces a valid
labeling and is safe with probability $>1-\frac 1n$, uses $\ge n+\sqrt
{\frac {n-1} 2} -1$ states. Hence, our protocol is almost
state-optimal. We also present a generalization of the protocol to
include a trade-off between the number of states and the expected
number of interactions. Finally, we show that for any silent and safe
labeling protocol utilizing $n+t<2n$ states the expected number of
interactions required to achieve a valid labeling is $\ge
\frac{n^2}{t+1}$.
\end{abstract}
\vfill
\newpage
\section{Introduction}
\label{section: Introduction}
The problem of assigning and further maintaining unique identifiers for
entities in distributed systems is one of the core problems related to
network integrity.  In addition, a solution to this problem is often
an important preprocessing step for more complex distributed
algorithms.  The tighter the range that the identifiers are drawn from,
the harder the assignment problem becomes.

In this paper we adopt the probabilistic population protocol model in
which we study the problem of assigning distinct
identifiers, which we refer to as {\em labels}, to all agents
\footnote{When the size of the label range is equal to the number of agents,
the problem is also called {\em ranking} in the literature \cite{BCC}}.
The adopted model was
originally intended to model large systems of agents with limited
resources (state space) \cite{AA}.  In this model the agents are
prompted to interact with one another towards a solution of a shared
task.  The execution of a protocol in this model is a sequence of
pairwise interactions between randomly chosen agents.  During an
interaction, each of the two agents: the {\em initiator} and the {\em
  responder} (the asymmetry assumed in \cite{AA}) updates its state in
response to the observed state of the other agent according to the
predefined (global) transition function.
For more details about
the population protocol model see Appendix~A.

Designing our population protocols for the problem of assigning unique
labels to the agents (labeling problem), we make an assumption that
the number $n$ of agents is known in advance.  Our protocols would
also work if only an upper bound on the number of agents is known to
agents.  In fact, in such case the problem becomes easier as the range
from which the labels are drawn is larger.  In particular, if we do
not have the limit on $n$ we also do not have limit on the number of
states to be used. More natural assumption is that such a limit is
imposed.  And indeed, there are plenty of population protocols which rely on
the knowledge of $n$ \cite{BCC,CIW}.

Our labeling
protocols include a preprocessing for
electing a {\em leader}, i.e., an agent
singled out from the population, which improves coordination of more
complex tasks and processes.  A good example is synchronization via phase clocks
propelled by leaders.  More examples of leader-based computation can
be found in~\cite{AAE}.

In the unique labeling problem adopted here, the number of utilized
states needs to reflect the number of agents $n.$
Also, $\Omega(n \log n)$ is a natural lower bound
on the expected number of interactions 
required to solve not only the labeling problem
but any non-trivial problem by a population protocol.
The main reason is that $\Omega(n \log n)$ interactions are needed to achieve a positive
constant probability that each agent is involved in at least one interaction \cite{BKR}.

Perhaps the simplest protocol for unique labeling in population networks is as follows \cite{CIW} (cf. \cite{BBS}). Initially, all agents hold label $1$ which is equivalent with all agents being in state $1.$  
In due course, whenever two agents with the same label $i$ interact, the
responder updates own label to $i+1.$ 
The advantage of this simple protocol is that it
does not need any
knowledge of the population size $n$ and it utilizes only $n$ states
and assigns labels from the smallest possible range $[1,n]$\footnote{We shall denote a range
  $[p,\cdots,q]$ by $[p,q]$ from here on.}.
The severe disadvantage is that it needs 
at least a cubic in $n$ number of interactions 
(getting rid of the last multiple label $i,$ for all $i=1,\dots\,n-1,$
requires a quadratic number of interactions in expectation) 
to achieve the configuration
in which the agents have distinct labels. 

In the following two examples of protocols for unique labeling,
we assume that the population size $n$ is 
embedded in the transition function,
such protocols are commonly used and known as  non-uniform protocols~\cite{AG},
and one of the agents is distinguished 
as the leader, see leader based protocols~\cite{AAE}.

In the first of the two examples, we instruct the leader to pass labels $n, n-1,..., 2$ to the encountered subsequently unlabeled yet agents and finally assign $1$ to itself. The protocol uses only $2n-1$ states ($n$ states utilized by the leader and $n-1$ states by other agents) and it
assigns unique labels in the smallest possible range $[1,n]$ to the $n$ agents.  
Unfortunately, this simple protocol requires $\Omega(n^2\log n)$ interactions
because as more agents get their labels, interactions
between the leader and agents without labels become less likely.
The probability of such an encounter drops
from $\frac 1n$ at the beginning to $\frac 1 {n(n-1)}$ at the end of
the process.

By using randomization, we can obtain a much faster simple protocol as
follows.
We let the leader to broadcast the number $n$ to all agents.
It requires $O(n \log n)$ interactions w.h.p.
\footnote{That is with the probability
  at least $1-\frac 1 {n^{\alpha}}$, where $\alpha \ge 1$
  and $n$ is the number of agents.}.
When an agent gets the number $n$, it
uniformly at random picks a
number in $[1,n^3]$ as its label. The probability that a given pair of
agents gets the same label is only $\frac 1 {n^3}.$ Hence, this
protocol assigns unique labels to the agents with probability at
least $1-\frac 1n.$ It requires only $O(n\log n)$ interactions
w.h.p. The drawback is that it uses $O(n^3)$ states and the large
range $[1,n^3].$
This method also needs a large number of random bits independent for each agent.

Besides the efficiency and population size aspects, there are
other deep differences between the three examples of labeling  protocols.
An agent in the first protocol never knows whether or not
it shares its label with other agents. This deficiency cannot happen in the
case of the second protocol  but it takes place in the third
protocol although with a small probability.

The labeling protocols presented in this paper
are {\em silent} and {\em safe}.
We say that a  (non-necessarily labeling) protocol is silent
if eventually each agent reaches
its final state and remains in it forever.
We say that a labeling
protocol is safe if
it never updates the label assigned to any single agent.
While the concept of a silent
population protocol is well established in the literature \cite{BCC,DGS99},
the concept of a safe labeling protocol is new.
The latter property is useful in the situation
when the protocol producing a valid labeling
has to be terminated before completion due to some
unexpected emergency or running out of time. 

Observe that among the three examples of labeling protocols, only the
second one is both silent and safe.  The first example protocol is silent \cite{BCC}
but not safe.  Finally, the third (probabilistic) one is silent and
almost safe as it violates
the definition only with small
probability.

\subsection{Our contributions}

The primary objective of this paper is to provide efficient labeling
protocols complemented with tight lower bounds in the aspects of the
number of states utilized by agents, the size of the range from which
the labels are drawn, and the expected number of interactions required
by our solutions.

In particular, we provide  positive answers
to two following natural questions
under the assumption that
the number $n$ of agents is known 
at the beginning.

\begin{enumerate}
\item
Can one design
a protocol for the
labeling problem requiring an asymptotically
optimal number of $O(n\log n)$ interactions w.h.p., 
utilizing an asymptotically optimal number of $O(n)$ states
and an asymptotically minimal label range of size $O(n)$ ?

\item Can one design a silent and safe protocol for
  the labeling problem utilizing substantially smaller
  number of 
  states than $2n$ and possibly the minimal label range $[1,n]$ ?
\end{enumerate}

We first present a population
protocol that w.h.p. requires an asymptotically
optimal number of $O(n \log n)$ interactions to assign distinct
labels from the range $[1,2n]$.  The protocol uses an asymptotically optimal number of
$O(n)$ states.
We also present a more involved generalization of the protocol, where
the range of assigned labels is $[1,(1+\varepsilon ) n]$. The generalized
protocol requires $O(n \log n / \varepsilon )$ interactions in order to
complete the assignment of distinct labels from $[1,(1+\varepsilon ) n]$
to the $n$ agents, w.h.p.  It uses $(2+\varepsilon)n+O(n^c)$
states, for any positive $c<1.$
Both protocols are silent w.h.p. and
safe.
Furthermore, we consider a natural class of population protocols for the
unique labeling problem, the so-called {\em pool protocols}, including
our fast labeling protocols.  We show that for any protocol in this
class that picks the labels from the range $[1,n+r]$, the expected
number of interactions is $\Omega(\frac {n^2}{r+1})$

Next, we provide a labeling protocol which uses only $n+5\sqrt n +O(n^c)$
states, for any positive $c<1,$ and the label range $[1,n]$. The expected number of interactions
required by the protocol is $O(n^3).$
Once a unique leader is elected it produces a valid labeling
and it is silent and safe.
On the other hand, we show that (even if a unique leader is given in advance)
any silent protocol that produces a valid labeling and is safe
with probability larger than $1-\frac 1n$,
uses at least $n+\sqrt {\frac {n-1} 2} -1$ states.  It follows
that our protocol is almost state-optimal.  In addition, we present a
variant of this protocol which uses $n(1+\varepsilon)+O(n^c)$ states,
for any positive $c<1.$  The
expected number of interactions required by this variation is
$O(n^2/\varepsilon^2),$ where $\varepsilon=\Omega(n^{-1/2}).$
On the other hand, we show that
for any silent and safe labeling protocol utilizing
$n+t<2n$ states the expected number
of interactions required
to achieve a valid labeling is at least $\frac{n^2}{t+1}$.

All our labeling protocols include a preprocessing for electing a
unique leader and assume the knowledge of the population
size $n.$ However, our almost state-optimal protocol (Single-Cycle protocol) can be made
independent of $n$ (see Section 4).

Our results are summarized in Tables 1 and 2.

\subsection{Main ideas of our protocols}

Our first fast labeling protocol roughly operates as follows.  The
leader initially has label $1$ and a range of labels $[2,n].$ During the execution
of the first phase, encountered unlabeled agents also get
a label and an interval of
labels that they can distribute among other agents. Upon a
communication between a labeled agent that has a non-empty interval
and an unlabeled agent, the latter agent gets a label from the
interval and if the remaining part of the interval has length $\ge 2$
then it is shared between the two agents.  After $O(n\log n)$
interactions, a sufficiently large fraction of agents is labeled and
has no additional labels to distribute w.h.p. The leader counts its
own interactions up to $O(\log n)$ in order to trigger the second
phase by broadcasting.  In the latter phase, an agent with a label $x$
and without a non-empty interval
can distribute one additional label $x+n.$ In the first phase of this
protocol the labels from the range $[1,n]$ are distributed rapidly
among the agents. In the second phase the unlabeled agent still have a
high chance of communicating with an agent that can distribute a
label. Roughly, our second, generalized fast labeling protocol is 
obtained from the first one by constraining the set of agents
that may distribute the labels $x+n$ in the second phase
to those having labels in the range $[1,n\epsilon ].$

The main idea of the almost state-optimal
labeling
protocol (Single-Cycle protocol) is to use the leader and an auxiliary
leader {\em nominated} by the leader to disperse the $n$ labels jointly
among the remaining free agents. The leader disperses the first and
the auxiliary leader the second part of each individual label. When a
free agent gets both partial labels, it combines them into its
individual label and then informs the leaders about this.  The two
leaders operate in two embedded loops. For each of roughly $\sqrt n$
partial labels of the leader, the auxiliary leader makes a full round
of dispersing its roughly $\sqrt n$ partial labels.  In the
generalized version of the protocol (k-Cycle protocol), the process is
partially parallelized by letting the leader to form $k$ pairs of
dispensers, where each pair labels agents in a distinct range of size
$n/k.$

\begin{table*}[t]
\begin{center}
\begin{tabular}{||c|c|c|c||} \hline \hline
Theorem & $\#$ states & $\#$ interactions & Range 
\\ \hline \hline
Theorem \ref{theo: 2n} & $O(n)$ & $O(n \log n)$ w.h.p.  & $[1,2n]$
\\ \hline
Theorem \ref{theo: epsilon} & $(2+\varepsilon)n+O(n^c)$, any $c<1$  & $O(n\log n /\varepsilon )$ w.h.p. &
$[1,(1+\varepsilon)n]$
\\ \hline
Theorem \ref{theo: scycle} &  $n+5\cdot\sqrt n+O(n^c)$, any $c<1$ & expected  $O(n^3)$ & $[1,n]$ 
\\ \hline
Theorem \ref{theo: kcycle} & $(1+\varepsilon)n+O(n^c)$, any $c<1$ & expected $O(n^2/\varepsilon^2)$ & $[1,n]$ 
\\ \hline \hline
\end{tabular}
\label{table: 1}
\vskip 0.5cm
\caption{Upper bounds on the number of states, the
  number of interactions and the range required by the
    labeling protocols
  presented in this paper.
  In Theorem \ref{theo: epsilon}, $\varepsilon$ is $\Omega (n^{-1})$
  while in Theorem \ref{theo: kcycle} $\Omega (n^{-0.5})$.}
\end{center}
\end{table*}

\begin{table*}[t]
\begin{center}
\begin{tabular}{||c|c|c|c||} \hline \hline
Protocol type & $\#$ states & $\#$ interactions & Theorem
\\ \hline \hline
any$^1$ &  $n$ & $\Omega(n\log n)$ w.h.p. & Theorem \ref{theo: trivial}
\\ \hline
silent, safe$^2$ & $n+\sqrt {\frac {n-1} 2} -1$ & - & Theorem \ref{theo: ssharp} (1st part)
\\ \hline
silent, safe$^3$, $n+t<2n$ states & - & expected $\frac {n^2}{t+1}$ & Theorem \ref{theo: ssharp} (2nd part)
\\ \hline
pool, range $[1,n+r]$ &  - & expected  $\frac {n^2}{r+1}$  & Theorem \ref{theo: last}
\\ \hline \hline
\end{tabular}
\label{table: 2}
\vskip 0.5cm
\caption{Lower bounds on the number of states or/and the
  number of interactions required by labeling protocols.
  (1) Any labeling protocol
  that is capable to produce a valid labeling.
 (2) The silent protocol in Theorem \ref{theo: ssharp} (first part)
 is assumed to produce a valid labeling and be safe with probability
 greater than $1-\frac 1n$.
(3) The silent protocol in Theorem \ref{theo: ssharp} (2nd part)
 is assumed to produce a valid labeling and be safe with probability $1$.
  }
\end{center}
\end{table*}

\subsection{Related work}

There are several papers concerning labeling of
processing units  (also known as renaming or naming) in different communication models
\cite{CRR11}.
E.g., Berenbrink et al. \cite{BBF} present
efficient algorithms for the so-called lose and tight renaming in
shared memory systems improving on or providing alternative
algorithms to the earlier
algorithms by Alistarh et al. \cite{ADR14,AAG10}.
The lose renaming
where the label space is larger that the number of units is shown
to admit substantially faster algorithms than the tight renaming
\cite{AAG10,BBF}.

The problem of assigning unique labels
to agents has been studied in the model of population protocols by
Beauquier et al. \cite{BBR,BBS}.
In \cite{BBS}, the
emphasis is on estimating the minimum number of states which are
required by apparently non-safe protocols. In \cite{BBR}, the authors
provide among other things a generalization of a leader election
protocol to include a distribution of $m$ labels among $n$ agents,
where $m \leq n$.  In the special case of $m = n$, all agents will
receive unique labels. No  analysis on the number of interactions
required by the protocol is provided in \cite{BBR}.
Their focus is on the feasibility of the solution, i.e., that the
process eventually stabilizes in the final
configuration. Their protocol seems inefficient in
the state space aspect as it needs many
states/bits to keep track of all the labels.

Doty et al.
considered the labeling problem in \cite{DEM} and presented
a subroutine named "UniqueID" for it
based on the technique of traversing a labeled binary tree and
associating agents with nodes in the tree. The subroutine requires
$O(n\log n \log \log n)$ interactions.

The labeling problem has also been studied in the context of
self-stabilizing protocols where the agents start in arbitrary (not
predefined) states, see \cite{BCC,CIW}.  In \cite{CIW}, Cai et
al. propose a solution which coincides with our first example of
labeling protocols presented in the introduction.  In a very recent
work \cite{BCC}, Burman et al. study both slow and fast labeling
protocols, the
latter utilizing an exponential number of states.  The protocols in both
papers require the exact knowledge of $n.$
The work \cite{BCC} focuses on self-stabilizing protocols which cannot
be safe by definition.  It is more proper to compare our protocols
with the initialized
version of the protocols in \cite{BCC}. E.g.,
the leader-driven initialized (silent) ranking protocol in \cite{BCC}
(see Lemma 4.1) requires $O(n^2)$ interactions, uses $O(n)$ states
and it is safe. An analogous variant of the fast ranking protocol
from  \cite{BCC} requiring $O(n\log n)$ interactions
and an exponential number of states is also safe but not silent.
The known labeling protocols are summarized in Table 3.
%in Appendix B.

\begin{table*}[h!]
  \begin{center}
\begin{tabular}{||c|c|c|c|c|c||} \hline \hline
$n$  &  $\#$ interactions & $\#$ states & Range & Properties & Paper
\\ \hline \hline
 unknown & $O(n^3)$ w.h.p. & n & $[1,n]$ & silent & \cite{CIW}
\\ \hline
unknown & $O(n\log n \log \log n)$ w.h.p. & $n^{O(1)}$ & $[1,n^{O(1)}]$ & silent & \cite{DEM}
\\ \hline
 known & $O(n^2)$ expected & $O(n)$ & $[1,n]$ & silent, safe & \cite{BCC}
\\ \hline
 known &  $O(n\log n)$ w.h.p. & $\exp (O(n^{\log n}\log n))$ & $[1,n]$ & safe & \cite{BCC}
\\ \hline \hline
\end{tabular}
\label{table: 3}
\vskip 0.5cm
\caption{Upper bounds on the number of interactions, the
  number of states and the range used by the known
  labeling protocols. In case of the self-stabilizing
  labeling protocols in \cite{BCC},
  the ``safe'' property
  can eventually hold only for their initialized versions.}
\end{center}
\end{table*}

The most closely related problem more studied in the literature
is that of counting the population size, i.e., the
number of agents. It has been recently studied by Aspnes et al.
in \cite{ABB16} and 
Berenbrink et al. in \cite{BKR}.
We assume
that the population size is initially known.
Alternatively, it can be computed by using the protocol
counting the exact population size given in \cite{BKR}.
The aforementioned protocol computes the population size
in $O(n \log n)$ interactions w.h.p.,
using $\tilde{O} (n)$ states.
Another possibility is to use
the protocol computing the approximate population size,
presented in \cite{BKR}. The latter  protocol
requires $O(n\log ^2 n)$ interactions to compute
the approximate size w.h.p.,
using only a poly-logarithmic number of
states. For references to earlier papers on protocols
for counting or estimating the population size,
in particular the papers that introduced the counting
problem and that include the original algorithms on which the improved
algorithms of Berenbrink et al. are based, see \cite{BKR}.

All our protocols include a preprocessing for electing a unique leader
and its synchronization with the proper labeling protocol (e.g., see
the proof of Theorem \ref{theo: 2n}). There is a vast literature
on population protocols for leader election \cite{BGK,DE,ER,GS20}.
For our purposes, the most relevant is the protocol that elects a
unique leader from a population of $n$ agents
using
$O(n \log n)$ interactions and $O(n^c)$ many states, for any positive
constant $c<1,$ w.h.p., described in \cite{BKR,DE} (see also Fact \ref{fact:
  fastleader}).
The newest results
elaborate on state-optimal leader election protocols utilizing 
$O(\log\log n)$ states. These include
the fastest possible protocol~\cite{BGK} based on 
$O(n\log n)$ interactions in expectation,
and a slightly slower protocol~\cite{GS20}
requiring $O(n\log^2 n)$ interactions with high probability.

Our population protocols for unique labeling use also the known
population protocol for (one-way) epidemics, or broadcasting.
It completes spreading a message in $\Theta (n \log n)$
interactions w.h.p. and it uses only two states \cite{ER}
(see also Fact \ref{fact: broad}).

\subsection{Organization of the paper}

In the next section, we provide basic facts on probabilistic
inequalities and population protocols for broadcasting,
counting and leader election. 
In Section 3, we present our fast silent w.h.p. and safe protocol for unique labeling
in the range $[1,2n]$ and its generalization to include the range $[1,n(1+\varepsilon)].$
Section 4
is devoted to the almost state-optimal, roughly
silent and safe
protocol
with the label range $[1,n]$ and its variation.
Section 5 presents lower bounds on the number of states 
or the number of interactions  for silent, safe
and the so-called pool protocols
for unique labeling. We conclude with Final remarks.

%%%%%%%%%%%%%%%%%%%%%%%%%%%%%%%%%%%%%%%%%%%%%%%%%%%%%%%%%%%%%%%%%%%%%%%%%%%

\section{Preliminaries}
\label{section: prelim}

\subsection{Probabilistic bounds}

\begin{fact} \label{fact: bound}(The union bound)
  For a sequence $A_1,\ A_2,...., A_r$ of events,
  $Prob(A_1\cup A_2\cup ......A_r)\le \sum_{i=1}^r Prob(A_i).$
\end{fact}

\begin{fact} \label{fact: chernoff}(multiplicative Chernoff lower bound)
  Suppose $X_1, ..., X_n$ are independent random variables
  taking values in $\{0, 1\}.$
  Let X denote their sum and let $\mu= E[X]$
  denote the sum's expected value. Then, for any $\delta \in [0,1],$
  \\
  $Prob(X\le (1-\delta)\mu)\le e^{\frac {\delta^2 \mu}2}$ holds.
  Similarly, for any $\delta \ge 0,$
  $Prob(X\le (1+\delta)\mu)\le e^{\frac {\delta^2 \mu}{2+\delta}}$ holds.  
  \end{fact}

\begin{fact}\label{fact: part}\cite{ER}
For all $C > 0$ and $0 < \delta < 1$,  during $C n\log n$ interactions,
with probability at least $1-n^{-O(\delta^2C)}$ , each agent participates in at least $2C(1-
 \delta) \log n$ and at most $2C(1+\delta) \log n$ interactions.
\end{fact}

\subsection{Broadcasting, counting and leader election}

We shall refer to the following  broadcast
process which can be completed during $\Theta(n\log n)$ interactions w.h.p.
Each agent is either in a state of M-type
(got the message) or in a state of $\neg$M-type.
Whenever an agent in a state of M-type interacts
with an agent in a state of $\neg$M-type, the latter changes its state  to a state of
M-type (gets the message).  The process starts when the first agent  gets the message and
completes when all agents have the message.

\begin{fact}\label{fact: broad}
  There is a constant $c_0$ , such that for $c \ge c_0,$ the broadcast process
  completes in $c n\log n$ interactions with probability at least $1-n^{-\Theta(c)}.$
\end{fact}

Berenbrink et al. \cite{BKR} obtained among other things the following results on
counting the population size, i.e., the number of agents.

\begin{fact}\label{fact: approx}
  There is a protocol for a population of an unknown number $n$
  of agents such that w.h.p., after $O(n \log^2 n)$ interactions
  the protocol stabilizes and each agent
  holds the same estimation of the population size which is either $\lceil \log n \rceil$
  or $\lfloor \log n \rfloor .$ The protocol uses $O(\log^2 n \log \log n)$
  states.
\end{fact}

\begin{fact}\label{fact: exact}
  There is a protocol for a population of an unknown number $n$
  of agents such that w.h.p., after $O(n \log n)$ interactions
  the protocol stabilizes and each agent
  holds the exact population size. The protocol uses $\tilde{O}(n)$
  states.
\end{fact}

There is a vast literature on population protocols for leader election
\cite{ER}.  For our purposes, the following fact will be sufficient.
Its idea is to start leader election with a subprotocol of \cite{GS20}
that elects a junta of substantially sublinear in $n$ number of
leaders. The junta is formed using $O(n\log n)$ interactions.  Then,
when state space of size $n^c$ is available, $c<1,$ only a constant
number of rounds of leader elimination is needed, each requiring
$O(n\log n)$ interactions. For more details, see
\cite{BKR,DE}.

\begin{fact}\label{fact: fastleader}
There is a protocol that elects a unique leader from a population of
$n$ agents
using
$O(n \log n)$ interactions and
$O(n^c)$  many states, for any positive
constant $c<1$, w.h.p. \cite{BKR,DE}.
\end{fact}

\section{Labeling with asymptotically optimal
  number of interactions, nearly optimal 
  number of states and range}
\label{section: range}
In this section, we provide a silent w.h.p. and
a safe labeling protocol
that assigns
unique labels from the range $[1,2n]$ to $n$ agents
in $O(n\log n)$ interactions w.h.p. Then, we generalize the protocol
to include the range $[1,(1+\varepsilon)n],$ where
$\varepsilon$ 
does not have to be a constant; it can even be as small
as $O(n^{-1})$. We show that
the generalized
protocol assigns unique labels from $[1,(1+\varepsilon)n]$
in $O(n\log n /\varepsilon)$ interactions w.h.p.
In the first  protocol,
the agents use $O(n)$ states, in the second protocol only $(2+\varepsilon)n+O(n^c)$
states, for any positive $c<1.$

\subsection{Range  $[1,2n]$}

The protocol runs in two main phases
preceded by a leader election preprocessing. The idea of the first phase
resembles that of load balancing \cite{BKR}, the difference
is that tokens (in our case labels and interval sub-ranges) are distinct.

\begin{figure}[h]
\begin{center}
\includegraphics[scale=0.5]{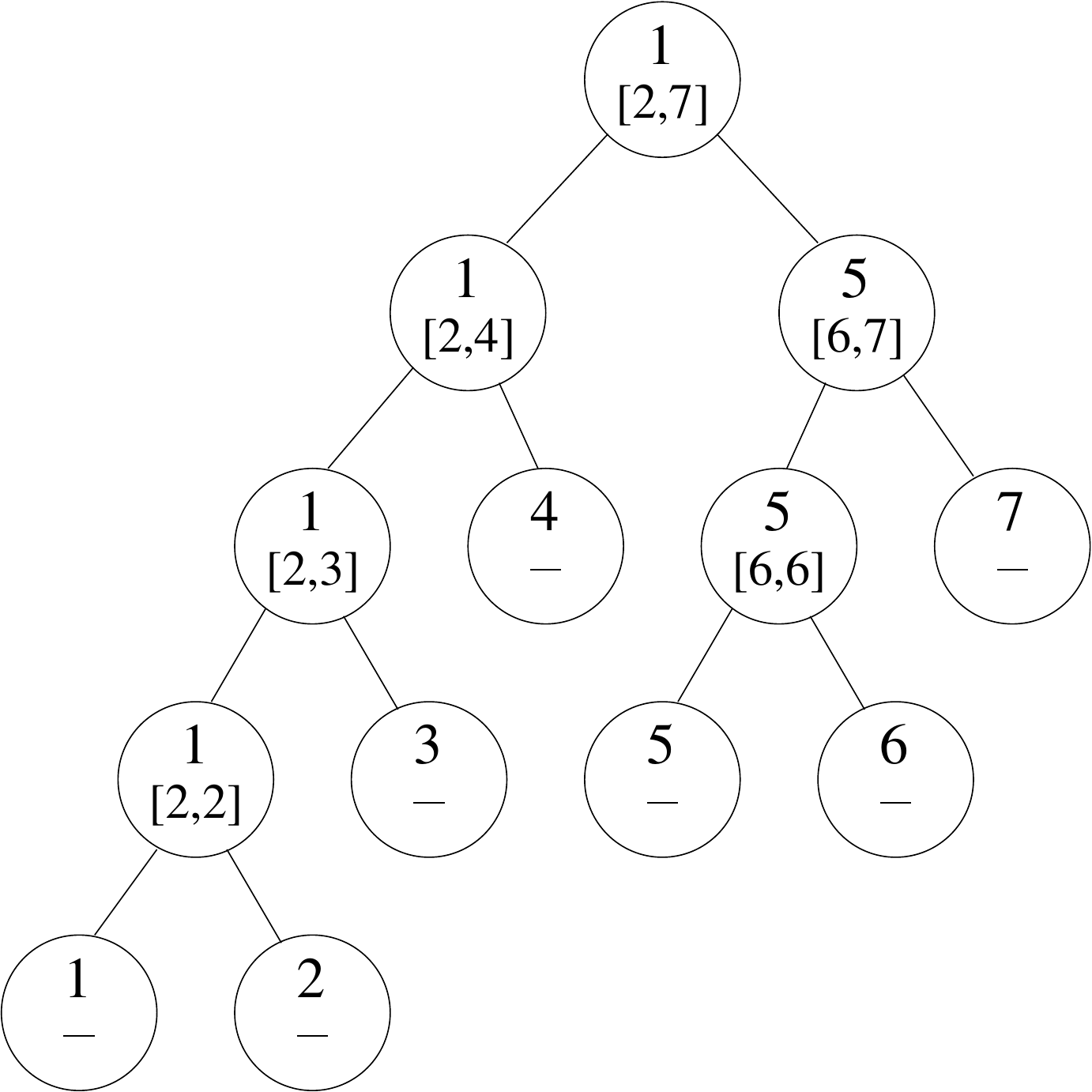}
\end{center}
\caption{An example of the partition tree of the start interval.}
\label{fig:interval}
\end{figure}

At the beginning of the first phase, the leader 
assigns the label $1$ and also temporarily the
interval $[2,n]$ to itself.  Next, whenever two agents
interact, one with label and a temporarily assigned interval $[q,r]$
where $r>q$
and the other without label, 
the former agent shrinks its interval to $[q,
  \lfloor \frac {q+r} 2 \rfloor ]$ and it gives away the label $\lfloor \frac {q+r} 2 \rfloor+1$
and if $\lfloor \frac {q+r} 2 \rfloor+2\le r$ also the sub-interval
$[\lfloor \frac {q+r} 2 \rfloor+2,r]$ to the latter agent.
Furthermore, whenever an agent with label and a temporarily assigned
singleton interval $[q,q]$ interacts with an agent without label,
the former agent cancels
its interval and gives the label $q$ to the latter
agent.
In the
remaining cases, interactions have no effect.
Note that during the first phase a sub-tree
of the binary tree of the partition  of
the start interval $[1,n]$ with $n$ leaves
determined by the protocol rules is formed, see Fig. \ref{fig:interval}.
%in Appendix C.
Also observe that when an agent at an intermediate node of
the tree interacts with an agent without label then the former
agent migrates to the left child of the node while
the latter agent lands at the right child of the node.

In the second phase,
when an agent with a label $i\in [1,n]$ at a leaf of the tree interacts with an
agent without label for the first time
then the latter agent gets the label $i+n.$ Interactions between agents (if any) at
intermediate nodes of the tree and agents without labels are defined as in the first phase.

The following lemmata are central in showing that $O(n\log n)$
interactions are sufficient w.h.p. to implement our protocol.

\begin{lemma}\label{lem: key}
There is a constant $c$ such that after
$c n\log n$ interactions in the first phase the number
of agents without labels  drops below $n/4$
w.h.p.
\end{lemma}
\begin{proof}
  The proof is by contradiction. Suppose that a set
  $F$ of at least $n/4$ agents without labels survives at least
  $c n\log n$ interactions, where the constant $c$ will be specified later.

  Consider first the leader agent starting with the interval
  $[2,n]$ during the aforementioned interactions.
  When the agent interacts with an agent without label its
  interval is roughly halved. We shall call such an interaction
  a success. The probability of success is at least
  $\frac 1 {4n}.$ The expected number of successes is at least
  $\frac c {4} \log n.$ By using Chernoff multiplicative
  bound given in Fact \ref{fact: chernoff}, we can set $c$ to enough large constant so
  the probability of at least $\log_2 n +1$ successes
  will be at least $1-\frac 1 {n^2}.$ This means that
  the leader will end up without
  any interval with so high probability during the
  $cn\log n$ interactions. The leader chooses
  the leftmost path in the binary partition tree 
  of the start interval $[1,n].$ Consider an arbitrary
  path $P$ from the root to a leaf in the tree.
  Note that several agents during distinct interactions
  can appear on the path. 
  Define as a success an interaction in which an agent
  currently on $P$ interacts with an agent without
  label. The expected number of successes is
  again  at least $\frac c {4} \log n$ and again we can conclude
  that there are at least $\log_2 n +1$  successes
  with probability  at least $1-\frac 1 {n^2}.$
  Simply, the probabilities of interacting with
  an agent without label are the same for
  all agents with labels, i.e., on some paths
  in the tree. Another way to argue is that the
  leader could make other decisions as
  to which roughly half of interval to
  preserve and the path choice.
  By the union bound (Fact \ref{fact: bound}), we conclude that all the
  $n$ paths from the root to the leaves in
  the tree could be developed during
  the $cn\log n$ interactions,
  so all agents would get a label,
  with probability at least $1-\frac 1n.$
  We obtain a contradiction with the
  so long existence of the set $F.$
 \qed \end{proof}

\begin{lemma} \label{lem: second}
  If the second phase starts after $c n\log n$ interactions,
  where $c$ is the constant from Lemma \ref{lem: key},
  then only $O(n\log n)$ interactions are needed to assign
  labels in $[1,2n]$ to the
  remaining agents without labels,   w.h.p.
\end{lemma}

\begin{proof}
  The number of agents without labels at the beginning of the second
  phase is at most $n/4$ w.h.p.  Hence, at the beginning of this
  phase the number of agents with labels is at least
  $\frac 3 {4}n$ w.h.p. An agent with label $i\le n$ at a leaf of the tree
  can give the label $i+n$ to an agent without
  label only once. Since this can happen 
  at most $\frac n {4}$ times,
  the number of agents with labels in $[1,n]$ that
  can give a label
  is always at least $\frac n 2$ w.h.p. We conclude that for an
  agent without label the probability of an interaction with an agent
  that can give a label is is at least almost $\frac 1 {2n}.$
  Hence, after each $O(n)$ interactions the expected number
  of agents without label halves. It follows that the expected
  number of such interactions rounds is $O(\log n).$ Consequently,
  the number of the rounds is also $O(\log n)$ w.h.p. by Chernoff bound
  (Fact \ref{fact: chernoff}).

  An alternative way to obtain the $O(n\log n)$ bound on the number of
  interactions w.h.p.  is to use Fact \ref{fact: part}
  with $C=O(\frac 1{1/2})$ and
  $\delta =\frac 12 .$ Then, each agent will interact
  with at least $C\log  n$ agents w.h.p.
  during $C n\log n$ interactions. Consequently, the
  probability that a given agent does not interact with any agent that
  can give a label during the aforementioned interactions is $(1-\frac
  1 2)^{O(2 \log n )}$. Hence, by picking enough large $C$,
  we conclude that each agent (in particular without label) will
  interact with at least one agent that can give a label 
  during the $C n\log n$ interactions w.h.p.
  \qed \end{proof}

\begin{lemma}\label{lem: unique}
 During both phases, no pair of agents
  gets the same label.
\end{lemma}
\begin{proof}
The uniqueness of the label assignments in
  the first phase follows from the disjointedness of the
  labels and intervals assigned to agents before and after each interaction.
  This argument also works for the labels not exceeding $n$
  assigned later in the second phase. Finally, the uniqueness
  of the labels of the form $i+n$ follows from the uniqueness
  of the labels of the agents passing these labels.
 \qed \end{proof}
  
\begin{theorem}\label{theo: 2n}
  There is a
  safe protocol
for population of $n$ agents that w.h.p. assigns
unique labels in the range $[1,2n]$ to the agents equipped
with $O(n)$ states
in $O(n\log n)$ interactions. The protocol is also silent w.h.p.
\end{theorem}
\begin{proof}
  Under the assumption that the leader election preprocessing
  provides a unique leader, the correctness of label assignment in both
  phases w.h.p. and the fulfilling of the definition of a silent and safe
  protocol follows from Lemmata \ref{lem: key},
  \ref{lem: second}, and \ref{lem: unique} and the specification
  of the protocol, respectively.

 For the purpose of the leader election preprocessing, we use the simple leader
election protocol using $O(n\log n)$ interactions and $O(n^c)$ states,
for any positive constant $c,$ described in \cite{BKR,DE}
(Fact \ref{fact: fastleader}).  The phase clock
(based on junta of leaders) from \cite{GS20} is also formed in
$O(n\log n)$ interactions, using $O(\log\log n)$ states and we use
this clock to count the required (by the simple leader election
protocol) time $\Omega(n\log n).$ When this time is reached on the
clock we switch from leader election to our proper labeling protocol.
The two aforementioned processes can be run simultaneously, resulting
in additional state usage $O(n^c \log \log n)$ (still fine for our needs).
Thus, the leader election preprocessing
  and its synchronization with the proper labeling protocol
  in two phases add $O(n\log n)$ interactions and $o(n)$ states w.h.p.
  It provides a unique leader w.h.p. It follows that w.h.p. the whole protocol
  provides a correct labeling, it is silent and safe. In fact, we can make it safe (with probability $1$)
  by prohibiting agents to change or get rid of an assigned label.
  Note that this constraint does not affect the operation of the
  protocol when a unique leader is provided by the preprocessing.
    
Both phases require $O(n\log n)$
interactions w.h.p. by Lemmata \ref{lem: key}, \ref{lem: second}.
  
To put the two phases described in Lemmata \ref{lem: key}, \ref{lem:
    second} together, we let the leader agent to count its
interactions.
  When the number of interactions of the leader in the first phase
  exceeds an appropriate multiplicity of $\log n$, the total number
  of interactions in the first phase achieves the required lower bound
  from Lemma \ref{lem: key} w.h.p. by Fact \ref{fact: part}.
  Therefore,
  then the leader starts broadcasting the message on the transition to the
  second phase to the other agents.  By Fact \ref{fact: broad},
  the broadcasting increases the number of interactions only by
  $O(n\log n)$ w.h.p.  (The leader can also stop the second phase in a
  similar fashion.)

  To save on the number of states, instead of having states
  corresponding to all possible sub-intervals of $[1,n],$
  we consider states corresponding to the nodes
  of the interval partition tree (see Fig. \ref{fig:interval})
  whose sub-tree is formed in
  the first phase. More precisely, we associate two states with
  each intermediate node of the binary tree on $n$ leaves and $n-1$
  intermediate nodes.
  They indicate whether or not the agent at the intermediate
  node has already received the message about the
  transition to the second phase.
  Next, we associate four states to each leaf of the tree.
  They indicate similarly whether or not the agent at the
  leaf has already received the phase transition message
  and whether or not the agent has already passed a label
  to an agent without label in the second phase, respectively.
With each label in the
  range $[n+1,2n],$ we associate only a single state.
    Additionally, there are $O(\log n)$ states used by the leader
    to count interactions in order to start
    the second phase. Recall also that
    the leader election preprocessing
    requires $o(n)$ additional states.
    Thus the total number of states
    does not exceed $2n+4n+n+o(n).$
  \qed \end{proof}

By combining the protocol of Theorem \ref{theo: 2n}
with that of  Berenbrink et al. for exact counting the population
size (Fact \ref{fact: exact}), we obtain
the following corollary on unique labeling when
the population size is unknown to agents initially.

\begin{corollary}\label{cor: 2n}
  There is a
  protocol
for a population
of $n$ agents that assigns
unique labels in the range $[1,2n]$ to the agents
initially not knowing the number $n,$ equipped
with $\tilde{O}(n)$ states,
in $O(n\log n)$ interactions w.h.p.
\end{corollary}
\begin{proof}
  We run first the protocol for exact counting (Fact \ref{fact: exact})
  and then
  our protocol for unique labeling (Theorem \ref{theo: 2n})
  using the leader elected by the counting protocol.
  We can synchronize the three protocols in a similar fashion
  as we synchronized the two phases of our protocol additionally
  using $O(n \log n)$ interactions and $O(\log n)$
  states.
 \qed \end{proof}

By using the method of approximate counting from \cite{BKR}
(Fact \ref{fact: approx})
instead of that for exact counting (Fact \ref{fact: exact}),
we can decrease the number of
states to $O(n)$
at the cost of increasing the label range to $[1,8n]$ and the
number of interactions required to $O(n \log^2 n).$

\subsection{Range  $[1,(1+\varepsilon)n]$}

The new protocol is obtained by
the following modifications in the previous one.
The leader which counts the number of own interactions
starts broadcasting the phase transition
message when the number of agents without
labels drops below $n\varepsilon /4$ w.h.p. (see Lemma \ref{lem: gkey}).
The information about the transition to the second
phase affects only the agents at the leaves of
the interval partition tree, corresponding to
labels in $[1,n\varepsilon ].$ When they get
the message about the phase transition, they
know that they can pass a label which is the sum
of their own label and $n$ to the first agent
without label they interact with. For this reason,
only the agents at the
leaves corresponding to labels in $[1,n\varepsilon]$
as well as the agents that are at the nodes that are
ancestors of the aforementioned leaves participate
in the broadcasting of the phase transition message.
(Observe that the number of agents at these ancestors is $O(n\varepsilon)$
and an agent at such an ancestor also has a label in $[1,n\varepsilon].$)
In the second phase, besides the agents at the leaves
corresponding to labels in $[1,n\varepsilon]$ and the agents without
labels, also the agents at the intermediate nodes of the tree
(if any) can really interact, in fact as in the first phase.

The following generalization of Lemma \ref{lem: key} is straightforward.
%see Appendix D for the proof.

\begin{lemma}\label{lem: gkey}
  Let $c$ be the constant from the statement of Lemma \ref{lem: key}.
  During $c n\log n/\varepsilon $ interactions in the first phase the number  of
  agents without label drops below $n \varepsilon/4$ w.h.p.
\end{lemma}
\begin{proof}
The proof is a generalization of that for
  Lemma \ref{lem: key}.
  Define $F_{\varepsilon}$ as a  set of 
at least $\varepsilon n/4$ agents without labels  that survive at least
$c n\log /\varepsilon $ interactions in the first phase.
Note that for an arbitrary agent, the probability of interaction
with a member in  $F_{\varepsilon}$ is at least $\frac {\varepsilon} {4n}$. 
The rest of the proof is analogous to that of Lemma \ref{lem: key}.
It is sufficient to replace $F$ by $F_{\varepsilon}$ and the
probability $\frac 1 {4n}$ of an interaction with a member
in $F$ with that  $\frac {\varepsilon} {4n}$ of an interaction
with a member in $F_{\varepsilon}$.
\qed \end{proof}

Having Lemma \ref{lem: gkey}, we can easily generalize
Lemma \ref{lem: second} to the following one.
%, see Appendix D for the proof.

\begin{lemma}\label{lem: gsecond}
  If the second phase starts after $c n\log n /\varepsilon $ interactions,
  where $c$ is the constant from Lemmata \ref{lem: key}, \ref{lem: gkey} ,
  then only $O(n\log n/ \varepsilon )$ interactions are needed
  to assign
  labels in $[1,(1+\varepsilon)n]$
  to the remaining agents without labels,
  w.h.p.
\end{lemma}
\begin{proof}
The number of agents without labels at the beginning of the second
  phase is smaller than $\varepsilon n/4$ w.h.p.
  Hence, at the beginning of the second
  phase the number of agents with labels in the range
  $[1,\varepsilon n]$ is at least
  $\frac {3\varepsilon n}  {4}$ w.h.p. Recall that
  such an agent at a leaf of the tree can give a label
  to an agent without label
  only once. It follows that the number of agents
  with labels in $[1,\varepsilon n]$ that can give a label
  to an agent without label is always at least $\frac {\varepsilon n}
  {2}$ w.h.p. We conclude that for an
  agent without label the probability of an interaction with an agent
  that can give a label is
  at least almost $\frac {\varepsilon}{2n}$.
  Hence, after each $O(n/\varepsilon)$ interactions the expected number
  of agents without labels halves.
  It follows that the expected
  number of such interactions rounds is $O(\log n).$ Consequently,
  the number of the rounds is also $O(\log n)$ w.h.p. by 
  Fact \ref{fact: chernoff}.

  An alternative way to obtain the $O(n\log n /\varepsilon )$ bound on the number of
  interactions w.h.p.  is to use Fact \ref{fact: part}
  analogously as in the proof of Lemma \ref{lem: second}.
  The difference is that $C$ is set to $O(\frac 2 {\varepsilon})$ instead
  of $O(2)$ since the set of agents that can give a label
  is of size at least $\frac {n\varepsilon} 2$ now.
 \qed \end{proof}

We also need the following auxiliary lemma on
broadcasting constrained to a subset of agents.
%see Appendix D for the proof.

\begin{lemma}\label{lem: broad}
  The leader can inform $\Theta (n\varepsilon)$ agents with labels not exceeding
  $O(n\varepsilon )$ about the phase transition using
  only these agents in $O(n \log n/ \varepsilon)$
  interactions.
\end{lemma}
\begin{proof}
During the initial part of the broadcasting process, after every
$O(n/\varepsilon )$ interactions, the expected number of agents
participating in the broadcasting process doubles. Hence, after
$O(n\log  /\varepsilon )$ interactions, the expected number of informed
agents will be $\Omega (n\varepsilon)$.  Then, the expected number of
uninformed agents will be halved for every $O(n/\varepsilon)$
interactions. So the expected number of rounds, each consisting of
$O(n/\varepsilon)$ interactions, needed to complete the broadcasting is
$O(\log n).$ It remains to turn the latter bound to a w.h.p. one.
This can be done by using the Chernoff bounds (Fact \ref{fact: chernoff}).

Alternatively, we can define for the purpose of the analysis
of the doubling part, a binary broadcast tree.
An informed
agent at an intermediate
node of the tree after
an interaction with an uninformed agent
moves to a child of the node while the other agent
now informed places at the other child (cf. the partition
tree in the proofs of Lemmata \ref{lem: key},
\ref{lem: gkey}).
Then, we can
use the technique from the proofs of Lemmata \ref{lem: key},
\ref{lem: gkey} to show
that only $O(n\log n/\varepsilon)$interactions are required w.h.p.
to achieve a configuration where only a constant fraction
of the agents participating in the broadcasting is uninformed.
To derive the same asymptotic upper bound on the number
of interactions required by the halving part w.h.p.,
we can use Fact \ref{fact: part} with $C=O(\varepsilon^{-1})$
analogously as in the proofs of Lemmata \ref{lem: second}, \ref{lem: gsecond}.
 \qed \end{proof}

The proof of the following theorem is analogous to that
of Theorem \ref{theo: 2n} with Lemmata \ref{lem: key},
\ref{lem: second} replaced by Lemmata \ref{lem: gkey}, \ref{lem: gsecond}.

\begin{theorem}\label{theo: epsilon}
  Let $\varepsilon >0.$
  There is a silent w.h.p. and
  safe protocol
for a population of $n$ agents that assigns
unique labels in the range $[1,(1+\varepsilon)n]$ to $n$ agents equipped
with $(2+\varepsilon)n+O(n^c)$ states, for any positive $c<1,$
in $O(n\log n /\varepsilon )$ interactions w.h.p.
\end{theorem}
  \begin{proof}
Under the assumption that the leader election preprocessing
    provides a unique leader, the correctness of the label assignment
    in both phases w.h.p.  and the fulfillment of the definition of a
    silent and safe protocol follow from Lemmata \ref{lem: gkey},
    \ref{lem: gsecond}, and \ref{lem: broad} by the
    same arguments as in the proof of Theorem \ref{theo: 2n}.

  The leader election preprocessing and its synchronization with
   the proper labeling protocol
  require $O(n\log n)$ interactions and 
  $o(n^c)$ states, for $c<1,$ w.h.p. as described in the proof of Theorem \ref{theo: 2n}.
  Analogously, it follows that w.h.p. the whole protocol provides
  a valid  labeling, it is silent and safe. Again, it can be transformed to a safe
  protocol by prohibiting agents to change or get rid of an assigned label.
  
  By Lemmata \ref{lem: gkey}, \ref{lem: gsecond},
  both phases require $O(n\log n  /\varepsilon )$ interactions w.h.p.
  The broadcasting about the phase transition starts when
  the number of agents without labels in the first phase drops
  below $n\varepsilon/4$ w.h.p. By Lemma \ref{lem: broad},
  it requires $O(n\log n/\varepsilon)$ interactions w.h.p. since only
  the $\Theta (n\varepsilon)$ agents in states corresponding to labels in
  $[1,n\varepsilon]$ are involved in it.
  
  The estimation of the number of needed states is more subtle than in
  Theorem \ref{theo: 2n}. With each intermediate node of the interval
  partition tree that does not correspond
  to a label in $[1,n\varepsilon]$
  (equivalently, that is not an ancestor
  of a leaf corresponding to a label in
  $[1,n\varepsilon]$), we associate a single state.
  (Recall here that if an agent at an intermediate
  node of the tree encounters an agent without
  label then the former agent moves to the left
  child of the node.)
      With each intermediate node corresponding
    to a label in $[1,n\varepsilon],$ we associate
    two states. They indicate whether or not
    the agent at the node has already got the message about
    phase transition. Next, with each leaf  of the tree
    corresponding to a label $i$ in $[1,n\varepsilon ],$
    we associate four states. They indicate whether
    or not the agent at the leaf has already got the message
    about the phase transition, and whether or not
    the agent has already passed the label $i+n$ to some
    agent without label, respectively.
  To each of the remaining
 leaves, we associate only a single state.

  We also need $O(\log n/\varepsilon)$ additional states for the leader to count the
  number of own interactions in order to start broadcasting the message on
  transition to phase two at a right time step. In fact, we can get
  rid of the $O(\frac 1 {\varepsilon})$ factor here by letting the leader
  to count approximately each $\Theta (1/e)$ interaction. Simply, the
  leader can count only interactions with agents which have got labels
  not exceeding $O(\varepsilon n).$

  Finally, we have $n\varepsilon$ states corresponding to the labels in $[n+1,(1+\varepsilon)n].$ 
   Thus,
   totally only $(2+O(\varepsilon))n+ O(n^c)$ states, for any positive
   $c<1,$ are
    sufficient. To get rid of the constant factor at $\varepsilon,$ it
    is sufficient to run the protocol for a smaller $\varepsilon'=\Omega
    (\varepsilon).$ It does not change the asymptotic upper bound on the
    number of required interactions w.h.p. and even it decreases the
    range of the labels.
 \qed \end{proof}

Note that $\varepsilon$ in Theorem \ref{theo: epsilon}
does not have to be a constant; it can even be as small
as $O(n^{-1}).$

By combining the protocol of Theorem \ref{theo: epsilon}
with that of  Berenbrink et al. for exact counting the population
size (Fact \ref{fact: exact}), we obtain
the following corollary on unique labeling when
the population size is unknown to agents initially.
The proof is analogous to that of Corollary \ref{cor: 2n}.

\begin{corollary}\label{cor: eps}
  Let $\varepsilon >0.$
  There is a
    protocol
for a population
of $n$ agents that assigns
unique labels in the range $[1,(1+\varepsilon)n]$ to the agents
initially not knowing the number $n,$ equipped
with $\tilde{O}(n)$ states
in $O(n\log n /\varepsilon )$ interactions w.h.p.
\end{corollary}

%%%%%%%%%%%%%%%%%%%%%%%%%%%%%%%%%%%%%%%%%%%%%%%%%%%%%%%%%%%%%%%%%%%%%%%%%%%%%%
\section{State- and range-optimal labeling}
In this section we propose and analyze
state-optimal
protocols,
which  are silent and safe once
a unique leader is elected, and 
utilize labels from the 
smallest possible range $[1,n].$
We assume
the number of agents $n$ 
to be known.
We propose such a labeling protocol {\em Single-Cycle} 
which utilizes $n+5\sqrt{n}+O(n^c)$ states, for any positive $c<1,$ and
the expected number
of interactions required by the protocol is $O(n^3).$
We show in Section 5  that any silent and safe labeling protocol
 requires $n+\sqrt{\frac {n-1}2}-1$ states, 
see Theorem \ref{theo: ssharp}. Thus,
our protocol is almost state-optimal.
Finally, we propose a partial 
parallelization of {\em Single-Cycle} protocol
called {\em $k$-Cycle} protocol which utilizes 
$(1+\varepsilon)n$ states and $O((n/\varepsilon)^2)$ interactions 
for $\varepsilon=\Omega(n^{-1/2}).$ 

\subsection{Labeling protocol}

The state efficient labeling protocol
starts from a preprocessing electing a unique leader.
Its main idea is to use two agents: 
the initial leader $A$ 
and a nominated (by $A$) agent $B,$ 
as partial {\em label dispensers}.
These two agents jointly dispense unique 
labels for the remaining {\em free} (non-labeled yet) 
agents in the population where 
agent $A$ dispenses the first and 
agent $B$ the second part of each individual label.
For the simplicity of presentation, we assume
that $n$ is a square of some integer.
During execution of the protocol 
agent $A$ uses partial labels 
${\tt label}(a)\in \{0,\dots,\sqrt n-1\}$ 
and $B$ uses partial labels 
${\tt label}(b)\in\{1,\dots,\sqrt n\}.$
The two dispensers label every agent by a unique pair of 
partial labels $({\tt label}(a),{\tt label}(b))$
where the combination $(i,j)$ is interpreted
as the integer label $i\cdot\sqrt{n}+j.$
The protocol first labels all {\em free} (different
to dispensers unlabeled) agents and eventually gives labels 
$(0,2)$ to agent $B$ and $(0,1)$ to agent  $A.$

In a nutshell, the labeling process is based on a single cycle of
interactions between dispensers $A$ and $B$ and the free agents. 
Agent $A$ awaits an interaction with a free agent $F$
when $A$ dispenses to $F$ its current partial label ${\tt label}(a).$ Now $F$ awaits an interaction with $B$ in order to receive the second part of its label. And when this happens agent $F$ concludes with the combined label and agent $B$
awaits an interaction with $A$ to inform that the next free agent needs to be labeled. On the conclusion of this interaction if ${\tt label}(b)>1$ agent $B$ adopts new partial label ${\tt label}(b)-1,$ otherwise $B$ adopts 
${\tt label}(b)=\sqrt n$ and agent $A$ adopts new label
${\tt label}(a)-1.$ The only exception is when 
${\tt label}(a)=0$ and ${\tt label}(b)=2$ when 
agent $B$ adopts label $(0,2)$ and agent $A$ adopts label $(0,1)$ and both agents conclude the labeling process.
%For more details, see the definition of the transition function in Appendix~E.
The state utilization and transition function
in the labeling protocol are specified as follows.
\par
\noindent
{\bf State utilization in} {\em Single-Cycle} {\bf protocol}

\vspace{0.2cm}
{\bf [Agent $A$]} Since ${\tt label}(a)\in\{0,\dots\sqrt n-1\}$
dispenser $A$ utilizes $2\cdot\sqrt n+2$ states including:
\begin{itemize} 
	\item $A.{\tt init}=(1)$ 
	{\small\bf  \#\ the initial (leadership) state of dispenser $A,$} 
	\item $A[{\tt label}(a),{\tt await}(F)]$ 
	{\small\bf  \#\ dispenser $A$ carrying partial label ${\tt label}(a)$ awaits interaction with a free agent $F,$}
	\item $A[{\tt label}(a),{\tt await}(B)]$ 
	{\small\bf  \#\ dispenser $A$ carrying partial label ${\tt label}(a)$ awaits interaction with dispenser $B,$}
	\item $A.{\tt final}=(0,1)$
	{\small\bf  \#\ the final state of $A.$}
\end{itemize}

{\bf [Agent $B$]} Since ${\tt label}(b)\in\{0,\dots\sqrt n\}$
dispenser $B$ utilizes  $2\cdot\sqrt n+3$ states including:
\begin{itemize} 
	\item $B[{\tt label}(b),{\tt await}(F)]$ 
	{\small\bf  \#\ dispenser $B$ carrying partial label ${\tt label}(b)$ awaits interaction with a free agent $F,$}
	\item $B[{\tt label}(b),{\tt await}(A)]$ 
	{\small\bf  \#\ dispenser $B$ carrying partial label ${\tt label}(b)$ awaits interaction with dispenser $A$}
	\item $B.{\tt final}=(0,2)$
	{\small\bf  \#\ the final state of $B.$}
\end{itemize}

{\bf [Agent $F$]} Since free agents carry 
partial labels ${\tt label}(a)\in\{0,\dots\sqrt n-1\}$ 
and eventually adopt one of the $n-2$ destination labels (excluding dispensers)
they utilize $n+\sqrt n-1$ states including:
\begin{itemize} 
	\item $F.{\tt init}=(0)$ 
	{\small\bf  \#\ the initial (non-leader) state of $F$} 
	\item $F[{\tt label}(a),{\tt await}(B)]$ 
	{\small\bf  \#\ free agent $F$ carrying partial label ${\tt label}(a)$ awaits interaction with dispenser $B,$}
	\item $F.{\tt final}=({\tt label}(a),{\tt label}(b))$
	{\small\bf  \#\ the final state of $F.$}
\end{itemize}
In total {\em Single-Cycle} protocol requires $n+5\cdot\sqrt n+4$ states.

\vspace*{0.3cm}
\noindent
{\bf Transition function in} {\em Single-Cycle} {\bf protocol}

\vspace*{0.2cm}
\noindent
{\bf [Step 0] Initialization}
During the first interaction of $A$ with a free agent the second dispenser $B$ is nominated.
Both dispensers adopt their largest labels.
Agent $A$ awaits a free agent in the initial state
while agent $B$ awaits a free agent carrying a partial 
label obtained from $A.$ 
\begin{itemize}
	\item $(A.{\tt init},F.{\tt init})\\
	\hspace*{1cm}\rightarrow (A[{\tt label}(a)=\sqrt n-1,{\tt await}(F)],
	B[{\tt label}(b)=\sqrt n,{\tt await}(F)]),$
	\\
\end{itemize}
\vspace*{-0.2cm}
The three steps $C_1, C_2,$ and $C_3$ of the labeling cycle are given below.

\noindent
{\bf [Step $C_1$] 
	Agent $A$ dispenses partial label}
During an interaction of agent $A$ with a free agent $F$
the current partial label ${\tt label}(a)$ is dispensed to $F$.
Both agents await interactions with dispenser $B$ which is
ready to interact with partially labeled $F$ but not $A.$
\begin{itemize}
	\item $(A[{\tt label}(a),{\tt await}(F)],
	F.{\tt init})\\
	\hspace*{1cm}\rightarrow(A[{\tt label}(a),{\tt await}(B)],
	F[{\tt label}(a),{\tt await}(B)])$ 
	{\small\bf  \#\ Go to Step $C_2$}
\end{itemize}
{\bf [Step $C_2$] 
	Agent $B$ dispenses partial label}
During an interaction of agent $B$ with a free agent $F$ which carries partial label ${\tt label}(a)$, the complementary current partial label ${\tt label}(b)$ is dispensed to $F$.
Agent $F$ concludes in the final state 
with the combined label $({\tt label}(a),{\tt label}(b)).$
Agent $B$ is now ready for interaction with $A.$
\begin{itemize}
	\item $(B[{\tt label}(b),{\tt await}(F)],
	F[{\tt label}(a),{\tt await}(B)])\\
	\hspace*{1cm}\rightarrow(B[{\tt label}(b),{\tt await}(A)],
	F.{\tt final}=({\tt label}(a),{\tt label}(b)))$
	{\small\bf  \#\ Go to Step $C_3$}
\end{itemize}
{\bf [Step $C_3$] 
	Agent $A$ and $B$ negotiate a new label or conclude}
In the case when ${\tt label}(a)=0$ and ${\tt label}(b)=2$ the dispensers $A$ and $B$
conclude in states $(0,1)$ and $(0,2)$ respectively,
see the first transition. 
Otherwise a new combination of partial labels is agreed
and the protocol goes back to Step $C_1.$
\begin{itemize}
	\item $(A[{\tt label}(a)=0,{\tt await}(B)],
	B[{\tt label}(b)=2,{\tt await}(A)])$\\
	\hspace*{1cm}$\rightarrow (A.{\tt final}=(0,1),
	B.{\tt final}=(0,2))$  
	{\small\bf  \#\ Conclude the labeling process}
	\item $(A[{\tt label}(a)=0,{\tt await}(B)],
	B[{\tt label}(b)>2,{\tt await}(A)])$ {\bf or}\\ 
	$(A[{\tt label}(a)>0,{\tt await}(B)],
	B[{\tt label}(b)>1,{\tt await}(A)])$\\
	\hspace*{1cm}$\rightarrow (A[{\tt label}(a),{\tt await}(F)], B[{\tt label}(b)-1,{\tt await}(F)])$  
	{\small\bf  \#\ Go to Step $C_1$}
	\item $(A[{\tt label}(a)>0,{\tt await}(B)],
	B[{\tt label}(b)=1,{\tt await}(A)])$\\
	\hspace*{1cm}$\rightarrow (A[{\tt label}(a)-1,{\tt await}(F)],
	B[{\tt label}(b)=\sqrt{n},{\tt await}(F)])$
	{\small\bf  \#\ Go to Step $C_1$}
\end{itemize}

\begin{theorem}\label{theo: scycle}
  {\em Single-cycle} utilizes $n+5\cdot\sqrt n+O(n^c)$ states, for any positive $c<1,$
  and the minimal label range $[1,n].$ The expected
  number of interactions required by the protocol is $O(n^3)$.
  Once a unique leader is elected, it produces a valid labeling
  of the $n$ agents and it is silent and safe.
\end{theorem}
\begin{proof}
  Assume that the leader election preprocessing
  provides a unique leader.
  Then, the protocol is silent and safe by its definition.
  All ll labels 
	are dispensed in the sequential manner
	and the labeling process concludes when the two 
	dispensers finalize their own labels.
	In particular, as soon as 
	the two dispensers $A$ and $B$ are established 
	they operate in a short cycle formed of steps $C1,C2$ and $C3$
	labeling one by one all free agents in the population.
	One can observe that the sequence of cycles mimics the structure 
	of two nested loops where the external loop iterates along
	the partial labels of $A$ and the internal one along
	partial labels of $B$. In total, we have $n-2$ iterations
	where
               the expected number of interactions required by each iteration is $O(n^2).$
	Thus one can conclude that
                the expected number of interactions required by 
	the whole labeling process but for the leader election preprocessing is $O(n^3).$ 
        By the definition
        of the protocol the range of assigned labels is $[1,n].$
	Finally, as indicated earlier in this section 
	the number of states utilized 
	by the protocol
        but for the leader election preprocessing is equal to $n+5\cdot\sqrt n+4$.

The leader election preprocessing and its synchronization with
   the proper labeling protocol
  require additional $O(n\log n)$ interactions and additional
  $o(n^c)$ states, for $c<1,$ w.h.p. as described in the proof of Theorem \ref{theo: 2n}.
  \qed \end{proof}
Observe that when the exact value of $n$ is embedded in
the transition function on the conclusion all agents 
become dormant, i.e., they stop participating 
in the labeling process.
One could redesign the protocol such that 
the labels are dispensed by $A$ and $B$ in the increasing order
using a diagonal method,
e.g., $(0,0),$ $(0,1),$ $(1,0),$
$(0,2),$ $(1,1),$ $(2,0),$ $(0,3),$ $(1,2),$ $(2,1),$ $(3,0)$ etc.,
where agent $A$ gets label $(0,0),$
agent $B$ gets label $(0,1),$
the first labeled free agent gets $(1,0),$
the second $(0,2),$ then $(1,1)$ and $(2,0),$
when $A$ and $B$ start using the next diagonal, etc.
Each pair $(i,j)$ is interpreted as $(i+j)(i+j+1)/2 + i$,
e.g., $(0,1)=1,$
$(0,2)=3,$ $(0,3)=6$ and in general $(0,j) = j(j+1)/2,$
$(1,j-1)=j(j+1)/2+1,$
$(1,j-2)=j(j+1)/2+2,$....,$(j,0)=j(j+1)/2+j$
$=(j+1)(j+2)/2-1= (0,j+1)-1.$
In this case the size of the population does not 
need to be known in advance, however, 
the two dispensers will never stop searching
for free agents yet to be labeled.

\subsection{Faster Labeling}

We observe that one can partially parallelize 
{\em Single-Cycle} protocol by
instructing leader $A$ to form $k$ pairs of dispensers where each pair labels agents in a distinct 
range of size $n/k.$
In such case the new {\em $k$-cycle} protocol 
requires extra $2k$ states
to allow leader $A$ initialize the labeling process 
(create two dispensers) in all $k$ cycles. 
Thus the total number of states is bounded by 
$n+2k+k\cdot(5\sqrt{n/k}+4)=n+6k+5k\cdot\sqrt{n/k}
<n+6(k+\sqrt{nk})<n+12\sqrt{nk},$ as $k<\sqrt{nk}$,
plus the number of states required by the leader election preprocessing.
We use the same method for the leader election
preprocessing and its synchronization
with the proper labeling protocol 
described in the proof of Theorem \ref{theo: 2n}. Analogously, it adds 
$O(n\log n)$ interactions and $O(n^c)$ states,
for any positive $c<1.$
As we need to pick $k$ for which $n+12\sqrt{nk}\le n+n\varepsilon$ we
conclude that $k\le n\varepsilon^2/144.$

One can show that for $k=n\varepsilon^2/144,$
the expected number of interactions required by the
{\em $k$-cycle} protocol is $O(n^2/\varepsilon^2).$
Note that in order to initialize $k$ cycles the leader $A$
has to communicate with $2k-1$ free agents. As $k$ is at most
a small
fraction of $n$ during the search for dispensers for each cycle
the number of free agents is always greater than $n/2$ 
(in fact it is very close to $n$). Thus the probability of forming a new dispenser during any interaction is greater than $1/2n$, i.e.,
the product of the probability $1/n$ that the random scheduler
selects leader $A$ as the initiator, times 
the probability greater than $1/2$ that the responder is a free agent.
In order to finish the initialization, we need to create new 
dispensers $2k-1$ times. Using Chernoff bound, we observe that after $O(kn)=O(n^2/\varepsilon^2)$ interactions all $k$ cycles have 
their two dispensers formed.
As each cycle dispenses $n/k=144/\varepsilon^2$ labels and the expected
number of interactions required
to dispense a single label is $O(n^2)$ with high probability,
the expected number of interactions
required by
a specific cycle to generate all labels is $O(n^2/\varepsilon^2)$ also with high probability.
As observed earlier,
the leader election preprocessing adds only $O(n\log n)$
interactions w.h.p.
Hence, the expected number of interactions required to
conclude the labeling process is $O(n^2/\varepsilon^2)$.
Finally, note that 
for small values of $\varepsilon$ approaching $n^{-1/2}$ {\em $k$-cycle} protocol reduces
to {\em Single-cycle} protocol and for constant $\varepsilon$ the number of interactions 
required by the protocol is $O(n^2).$
\begin{theorem}\label{theo: kcycle}
For $k=n\varepsilon^2/144,$ where $\varepsilon=\Omega(n^{-1/2}),$ 
and the minimal label range $[1,n],$
the proposed
{\em $k$-cycle} labeling protocol provides a space-time trade-off
in which utilization of $(1+\varepsilon)n+O(\log \log n)$ states permits
the expected number of interactions $O(n^2/\varepsilon^2)$.
\end{theorem}

%%%%%%%%%%%%%%%%%%%%%%%%%%%%%%%%%%%%%%%%%%%%%%%%%%%%%%%%%%%%%%%%%%%%%%%%%%%%%%

\section{Lower bounds}

In this chapter, we derive several lower bounds on the number of
states or the number of interactions required by silent, safe or
the
so-called pool protocols for unique labeling. Importantly, these lower
bounds also hold in our model assuming that the population size is
known to the agents initially and also when a unique leader is available
initially.

The following general lower bound valid for any range of labels
follows immediately from the definitions of a population
protocol and the problem of unique labeling, respectively.

\begin{theorem}\label{theo: trivial}
  The problem of assigning unique labels to $n$ agents
  requires $\Omega (n\log n)$ interactions w.h.p.
  and the agents have to be equipped with at least
  $n$ states.
\end{theorem}
\begin{proof}
  $\Omega (n\log n)$ interactions are needed w.h.p. since each
  agent has to interact at least once, see, e.g., the introduction in \cite{BKR}.
  The lower bound on the number of states follows from
  the symmetry of agents, so any agent
  has to be prepared
  to be assigned an arbitrary label with at least a logarithmic
  bit representation.
   \qed \end{proof}

\subsection{A sharper lower bound on the number of states}

We obtain the following lower bound
on the number of states required by a silent protocol
which produces a valid labeling of the $n$
agents and is safe 
w.h.p. The lower bound holds even if
the protocol is provided with a unique leader
and the knowledge of the number of agents.
It almost matches the upper bound established
in the previous section.

\begin{theorem}\label{theo: ssharp}
A silent protocol which produces
a valid labeling of the $n$ agents and
is safe 
with probability larger than
$1-\frac 1n $ requires at least $n+\sqrt {\frac {n-1} 2} -1$ states.
Also, if a silent protocol,
which produces
a valid labeling of the $n$ agents and is safe 
with probability $1$, uses $n+t$ states, where $t<n,$ then
the expected number of interactions required by the protocol
to provide a valid labeling
is $\frac {n^2}{t+1}$.
\end{theorem}
\begin{proof}
 Let $I$ be the set of ordered pairs of the $n$ agents.
  $I$ can be interpreted as the set of possible pairwise interactions
  between the agents.
  
  Let $Z$ be a finite  run of the protocol,
  i.e., a finite sequence of pairs in $I.$
  Suppose that the execution of $Z$ is successful, i.e.,  
  each agent reaches a final state with a distinct
  label, and no agent gets assigned two or more distinct labels
  during the run. 

Let $F_Z$ be the set of final states achieved by the agents
after the execution of the run $Z$. We
have $|F_Z|=n.$ Also, let $R_Z$ stand for the set of remaining states
used in this run. Observe that if an agent is in a state in $F_Z$
then it has a label.

For an agent $x,$ let $f_Z(x)\in F_Z$ be the last state achieved by the
agent in the run $Z,$ and let $pred_Z(x)$ be the next to the last state
achieved by the agent $x$ in the run. Since for at most one agent
the common initial state can be the final one,
$pred_Z(\ )$ is defined for at least $n-1$ agents.
If $pred_Z(x)\in F_Z$ and $pred_Z(x)$ assigns a distinct label from
that assigned by $f_Z(x)$ to $x$ then we have
a contradiction with our assumptions on $Z.$ In turn,
if $pred_Z(x)\in F_Z$ and $pred_Z(x)$ assign the same 
label as that assigned by $f_Z(x)$ to $x$ then we have
a contradiction with the validity of the final
labeling resulting from $Z.$
We conclude that if $pred_Z(x)$ is defined
then $pred_Z(x)\in R_Z.$

Next, let $A_Z$ be the set of agents $x$ that achieved their final
state in the run $Z$ by an interaction of $x$ in the state $pred_Z(x)$
with an agent in
a state in $F_Z.$
%For the proof of the following
%claim under the assumptions of the first statement in the theorem see Appendix F.

Under the theorem assumptions, we claim the following.

\begin{claim}\label{claim: first}
There is a
finite run $Z$ of the protocol such that after the execution of $Z$,
each agent is in a final state with a distinct label,
no single agent is assigned distinct labels during $Z,$
and for any pair
of distinct agents $x,\ y \in A_Z,$ $pred_Z(x)\neq pred_Z(y).$
\end{claim}
\par
\noindent
{\em Proof of Claim \ref{claim: first}.}
The proof of the claim is by a contradiction
with the assumptions on the labeling protocol. The general intuition is
that if $pred_Z(x)=pred_Z(y)$ for two agents $x,\ y \in A_Z$ then we can
associate with a prefix of $Z$ a slightly modified equally likely run
$Z'$ which assigns the same label to a pair of agents.
Hence, the modified run does not produce a valid labeling or
it has to assign at least two different labels to some agent.

To obtain the contradiction, we assume that for each finite run $Z$ in
which the agents achieve final states with distinct labels
without assigning distinct labels
to any single agent during the run, there is a
pair of agents $x,\ y \in A_Z$, where $pred_Z(x)=pred_Z(y).$ Let us
consider such a pair of agents $x, y \in A_Z$ that minimizes the
length of the prefix of $Z$ in which both agents achieve their final
states in $F_Z$. We may assume w.l.o.g. that $x$ gets its final state
$f_Z(x)$ in an interaction $i_1$ with an agent $x'$ that
is in a state  in $F_Z,$
and in a later interaction $i_2,$
$y$ gets its final state $f_Z(y)$, in the run $Z.$
(Note that $x'$ cannot be in a final state
different from its own, i.e., in $F_Z\setminus \{f_Z(x')\}$
since this would require updating its label
contradicting the assumption on $Z.$) Thus, the shortest
prefix of $Z$ in which both $x$ and $y$ get their final states has the
form $Z_1i_1Z_2i_2.$ Then, if we replace the latter interaction $i_2$
by the interaction $i_3$ between $y$ and the agent $x'$ in the state
$f_Z(x')$ analogous to $i_1$, it will result in achieving by $y$ the
state $f_Z(x)$ since $pred_Z(x)=pred_Z(y).$ Thus, neither the run
$Z_1i_1Z_2i_3$ nor any of its extensions
can yield a valid labeling of the
agents without updating
labels for some of them.  Importantly, the runs $Z_1i_1Z_2i_2$ and $Z_1i_1Z_2i_3$ are
equally likely (*).

We initialize two sets $S_{valid}$ and $S_{invalid}$ of
strings (sequences) over the alphabet $I.$ Then, for each run $Z$ in
which the agents achieve final states with distinct labels
without updating the label of any single agent, we insert
the prefix $Z_1i_1Z_2i_2$ into $S_{valid}$ and the corresponding
sequence $Z_1i_1Z_2i_3$ into $S_{invalid}$.
Note that by the choice of $i_1,\ i_2,$ no string in
$S_{valid}$ is a prefix of another string in $S_{valid}.$
The analogous property holds for $S_{invalid}.$
By the construction of
the sets, each run $Z$ in which the agents achieve final states with
distinct labels without updating the label
of any single agent has to overlap with or be a lengthening of
a string in $S_{valid}.$
Furthermore, no run of the
protocol that overlaps with  a string in
$S_{invalid}$ or it is a lengthening of a string
in $S_{invalid}$ results in a valid
labeling without updating the label
of any single agent. Define the function $g: S_{valid} \rightarrow S_{invalid}$
by $g(Z_1i_1Z_2i_2)=Z_1i_1Z_2i_3.$ By the property (*),
the probability that a string over $I$ is equal
to $Z_1i_1Z_2i_2$ or it is a lengthening of $Z_1i_1Z_2i_2$
is not greater than the probability that a string over $I$ is equal
to $g(Z_1i_1Z_2i_2)$ or it is a lengthening of $g(Z_1i_1Z_2i_2)$.
The function $g$ is not necessarily a bijection.
Suppose that $g(Z_1i_1Z_2i_2)=g(Z_1'i_1'Z_2'i_2').$
Then, we have $Z_1i_1Z_2i_3=Z_1'i_1'Z_2'i_3.$
Consequently, the strings $Z_1i_1Z_2i_2$ and $Z_1'i_1'Z_2'i'_2$
may only differ in the last interaction, i.e., $i_2$ may
be different from $i_2'.$ However, $i_2$ and $i_2'$
have to include the same agent ($y$ in the earlier construction)
that appears in $i_3.$ We conclude that the aforementioned
two strings in $S_{valid}$ can differ by at most one
agent in the last interaction. It follows that $g$ maps
at most $n-1$ strings in $S_{valid}$ to the same string
in $S_{invalid}.$ Consequently, the event that the
agents eventually achieve their final states yielding
a valid labeling without updating the label
of any single agent is at most $n-1$ times more likely
than the complement event.  We obtain a contradiction
with theorem assumptions which completes the proof
of Claim \ref{claim: first}.

From here on, we assume that the run $Z$ satisfies
the claim. Consequently,
$|R_Z|\ge |A_Z|$.

Let $B_Z$ be the set of remaining agents that got their final state
in $F_Z$ in an interaction where both agents were in states
outside $F_Z,$ i.e., in $R_Z.$ Since the agents in $B$ achieved
distinct final states with distinct labels
in the aforementioned interactions, we infer
that $2|R_Z|^2\ge |B_Z|$ and thus $|R_Z|\ge \sqrt {|B_Z|/2}$.
Simply, there are $|R_Z|^2$ ordered pairs of states in $R_Z,$
and when agents in the states forming such a pair
interact they can achieve at most two distinct states in $F_Z.$
(Consequently, if $2|R_Z|^2< |B_Z|$ then there would be a pair
of agents in $B_Z$ that would achieve the same final state in the run
and hence it would have the same label at the end of the
considered run.)

Thus, we obtain $|R_Z|\ge \max \{ |A_Z|, \sqrt {\frac {n-1-|A_Z|)} 2} \} \ge \sqrt {\frac {n-1}2} -1$
by straightforward calculations. This completes the proof of the first
statement of the theorem.

To prove the second statement of the theorem,
we need $|R_Z|\ge |A_Z|$ to hold for any run $Z$ resulting in a valid
labeling of the agents without updating the label
of any single agent. The existence of such a run $Z$ implied by
Claim \ref{claim: first} is not sufficient to obtain
a lower bound on the expected number of required interactions.
The stronger assumptions on the silent protocol in the second
statement of the theorem requiring the protocol
to provide always a valid labeling
without updating the label
of any single agent solves the problem. Namely,
if $pred_Z(x)=pred_Z(y)$ for $x,\ y \in A_Z$ then
following the notation and argumentation from the proof of
Claim \ref{claim: first} neither
$Z_1i_1Z_2i_3$ nor any of its lengthening can provide a valid labeling
without updating the label
of any single agent.
We obtain a contradiction with the aforementioned assumptions.
Thus, the inequality  $|R_Z|\ge |A_Z|$ holds for arbitrary
run $Z$ ending with a valid labeling without updating the label
of any single agent.

To prove the second statement, we may also assume w.l.o.g. that $|A_Z|<n$
since otherwise $t\ge |R_Z|\ge  |A_Z|\ge n.$ Hence, the set $B_Z$ of agents
is non-empty. Let $x$ be a last agent in $B_Z$ that being in
the state $pred(x)$ gets
its final state $f(x)$ by an interaction with another agent $y$
in a state $s.$ If $y$ belongs to $B_Z$ then both $x$ and $y$
are the two last agents in $B_Z$ that simultaneously get their
final states in $F_Z$ in the same interaction. The probability
of the interaction between them is only $\frac 1 {n^2}.$
Suppose in turn that $y$ belongs to $A_Z.$
We know that $t\ge |R_Z|\ge |A_Z|$ from the previous part.
Thus, there are at most $t$ agents in $B_Z$ in the state
$s$ with which the agent $x$ in the state $pred_Z(x)$ could interact.
The probability of such an interaction is at most $\frac {t} {n^2}$.
We conclude that the probability of an interaction
between the agent $x$ and the agent $y$ after which $x$ gets
its final state $f(x)$ is at most $\frac {t+1} {n^2}$ which
proves the second statement.
\qed \end{proof}

\begin{corollary}
  If for
  $\varepsilon > 0$, a
  silent protocol that
  produces a valid labeling of the
  $n$ agents and is safe
  with probability $1$
 uses only $n+O(n^{1-\varepsilon})$ states then
 the expected number of interactions required by the protocol
 to achieve a valid labeling is
 $\Omega (n^{1+\varepsilon}).$
 \end{corollary}

\subsection{A lower bound for the range $[1,n+r]$}

  Our fast protocols presented in Section 3 are examples
  of a class of natural protocols for the unique labeling
  problem that we term {\em pool protocols}.

    In each step of a pool protocol a subset of agents owns explicit
  or implicit pools of labels which are pairwise disjoint and whose
  union is included in the assumed range of labels.  When two agents
  interact, they can repartition the union of their pools among themselves.
  Before the start of a pool protocol, only a single agent (the
  leader) owns a pool of labels. This initial pool corresponds to the
  assumed range of labels. An agent can be assigned a label
  from its own pool only. After that, the label is removed from the pool
  and cannot be changed.
  Finally, an agent without assigned label cannot give away
  the whole own pool during  an interaction with another agent
  without getting some part of the pool belonging to the other agent.
 
\begin{theorem}\label{theo: last}
   The expected number of interactions required by a pool
   protocol to assign unique labels in the range $[1,n+r]$, where $r\ge 0,$  to the population of
   $n$ agents is at least $\frac {n^2}{r+1}.$
   \end{theorem}
   \begin{proof}
   We shall say that an agent has the $P$ property if the agent owns a
   non-empty pool or a label has been assigned to the agent. Observe
   that if an agent accomplishes the $P$ property during running a pool
   protocol then it never loses it. Also, all agents have to
   accomplish the $P$ property sooner or later in order to complete
   the assignment task. During each interaction of
   a pool protocol
   at most one more agent can get the $P$ property. Since at the beginning
   only one agent has the $P$ property, there must exist an interaction
   after which only one agent lacks this property. By the disjointedness
   of the pools and labels, the assumed label range,
   and the definition of a pool protocol,
   there are at most $r+1$ agents among the remaining ones that could donate
   a sub-pool or label from own pool to the agent missing the $P$ property.
   The expected number of interactions leading to an
   interaction between the agent missing the $P$ property
   and one of the at most $r+1$ agents is $\frac {n^2}{r+1}.$
  \qed \end{proof}

   \section{Final remarks}
Our upper bound of $n+5\cdot\sqrt n+O(n^c)$,
for any positive $c<1,$ on the number of
states achieved by a protocol 
for unique labeling
that is silent and safe once a unique leader is elected
almost matches
our lower bound of $n+\sqrt {\frac {n-1} 2} -1$.
 We can combine our protocols for unique labeling with the recent
 protocols for counting or approximating the population size due to
 Berenbrink et al. \cite{BKR} in order to get rid of the assumption
 that the population size is known to one of the agents initially.
 Since the aforementioned protocols from \cite{BKR} either require
 $\tilde{O}(n)$ states or $O(n\log^2 n)$ interactions, the resulting
 combinations lose some of the near-optimality or optimality
 properties of our protocols (cf. Corollaries \ref{cor: 2n},
 \ref{cor: eps}).  The related question
 if one can design a protocol for counting or closely approximating the
 population size simultaneously requiring $O(n \log n)$ interactions
 w.h.p. and at most $cn$ states, where $c$ is a low constant, is of
 interest in its own right.
%%%%%%%%%%%%%%%%%%%%%%%%%%%%%%%%%%%%%%%%%%%%%%%%%%%%%%%%%%%%%%%%%%%%%%%%%%%%%%

\bibliography{identities-lipics}

\appendix
\section{The computational model of population protocols}
There is given a population of $n$ agents that can pairwise interact in
order to change their states and in this way perform a computation.  A
population protocol can be formally specified by providing a set $Q$
of possible states, a set $O$ of possible outputs, a transition
function $\delta : Q \times Q \rightarrow Q \times Q$, and an output
function $o : Q \rightarrow O.$ The current state $q\in Q$ of an agent
is updated during interactions. Consequently, the current output
$o(q)$ of the agent also becomes updated during interactions.
The current state of the set of $n$ agents is given by
a vector in $Q^n$ with the current states of the  agents.
A computation of a population protocol is specified by a sequence
of pairwise interactions between agents.
In every time step, an ordered  pair of agents is selected for interaction
by a probabilistic scheduler independently and uniformly at random.
The first agent in the selected pair is called the initiator while the
second one is called the responder. The states of the two agents are
updated during the interaction according the transition function
$\delta .$

We can specify a problem to solve by a population protocol
by providing the set of input
configurations, the set $O$ of possible outputs, and the desired output configurations for given
input configurations. For the unique labeling problem, all agents are initially in the
same state $q_0$.
The set $O$ is just 
the set of positive integers. A desired configuration is
when all agents output their distinct labels.
 The {\em stabilization time} of an execution of a protocol is the number
 of interactions until the states of agents form a desired
configuration from which
no sequence of pairwise interactions can lead to a configuration outside the set of desired
configurations.
 \vfill
 \end{document}